\documentclass{amsart}
\usepackage[foot]{amsaddr}

\usepackage[T1]{fontenc}        

\usepackage{amsmath}
\usepackage{amsfonts}
\usepackage{mathtools}
\usepackage{amssymb}

\usepackage{complexity}

\theoremstyle{definition}
\newtheorem{definition}{Definition}

\theoremstyle{plain}
\newtheorem{theorem}{Theorem}
\newtheorem{corollary}[theorem]{Corollary}

\newtheorem{lemma}[theorem]{Lemma}

\DeclareSymbolFont{rsfscript}{OMS}{rsfs}{m}{n}
\DeclareSymbolFontAlphabet{\mathrsfs}{rsfscript}
\DeclareMathOperator{\Ker}{Ker}

\usepackage{pstricks,pst-node,pst-text,pst-3d}
\usepackage{color}

\newcommand{\mA}{\mathrsfs{A}}
\newcommand{\mB}{\mathrsfs{B}}
\newcommand{\mC}{\mathrsfs{C}}
\newcommand{\mE}{\mathrsfs{E}}

\newcommand{\mF}{\mathrsfs{F}}
\newcommand{\gR}{\mathrel{\mathfrak{R}}}
\newcommand{\gL}{\mathrel{\mathfrak{L}}}
\newcommand{\gD}{\mathrel{\mathfrak{D}}}

\newcommand{\UWS}{\ensuremath{\mathsf{UWS}}}

\newcommand{\sa}{synchronizing automata}
\newcommand{\san}{synchronizing automaton}
\newcommand{\sDFAs}{synchronizing DFAs}
\newcommand{\sDFA}{synchronizing DFA}

\newcommand{\sws}{reset words}

\usepackage{pgfplots}
\pgfplotsset{compat=1.16}

\usetikzlibrary{shapes,shapes.geometric,arrows,fit,calc,positioning,automata,shadows,arrows.meta,decorations.pathreplacing}

\newcommand{\ra}[3][black]{
\draw [line width=0.5mm, #1] (#2 + 0.7,#3 + 0.5) -- (#2 + 0.3,#3 + 0.75);
\draw [line width=0.5mm, #1] (#2 + 0.7,#3 + 0.5) -- (#2 + 0.3,#3 + 0.25);
}

\newcommand{\la}[3][black]{
\draw [line width=0.5mm, #1] (#2 + 0.3,#3 + 0.5) -- (#2 + 0.7,#3 + 0.75);
\draw [line width=0.5mm, #1] (#2 + 0.3,#3 + 0.5) -- (#2 + 0.7,#3 + 0.25);
}

\newcommand{\ua}[3][black]{
\draw [line width=0.5mm, #1] (#2 + 0.5,#3 + 0.7) -- (#2 + 0.75,#3 + 0.3);
\draw [line width=0.5mm, #1] (#2 + 0.5,#3 + 0.7) -- (#2 + 0.25,#3 + 0.3);
}

\newcommand{\da}[3][black]{
\draw [line width=0.5mm, #1] (#2 + 0.5,#3 + 0.3) -- (#2 + 0.75,#3 + 0.7);
\draw [line width=0.5mm, #1] (#2 + 0.5,#3 + 0.3) -- (#2 + 0.25,#3 + 0.7);
}

\definecolor{DarkRed}{RGB}{192,0,0}

\makeatletter

\renewcommand*\subjclass[2][2020]{\def\@subjclass{#2}\@ifundefined{subjclassname@#1}{\ClassWarning{\@classname}{Unknown edition (#1) of Mathematics Subject Classification; using '2020'.}}{\@xp\let\@xp\subjclassname\csname subjclassname@#1\endcsname}}

\renewcommand{\subjclassname}{\textup{2020} Mathematics Subject Classification}

\begin{document}

\title[Winning strategies for synchronization games]{Uniform winning strategies\\ for the synchronization games\\ on subclasses of finite automata}

\author{Henning Fernau}
\address[H. Fernau, C. Haase]{\normalfont Universit\"at Trier, Fachbereich IV, Informatikwissenschaften, Trier, Germany}
\email{fernau@uni-trier.de, haasec@uni-trier.de}
\author{Carolina Haase}
\author{Stefan Hoffmann}
\address[S. Hoffmann]{\normalfont Erfurt, Germany}
\email{hoffmanns.tcs@gmail.com}
\author{Mikhail Volkov}
\address[M. Volkov]{\normalfont Yekaterinburg, Russia} \email{m.v.volkov@urfu.ru}

\begin{abstract}
The pseudovariety $\mathbf{DS}$ consists of all finite monoids whose regular $\gD$-classes form subsemigroups. We exhibit a uniform winning strategy for Synchronizer in the synchronization game on every \san\ whose transition monoid lies in $\mathbf{DS}$, and we prove that $\mathbf{DS}$ is the largest pseudovariety with this property.
\end{abstract}

\maketitle

\section{Introduction}
\label{intro}

A complete \emph{deterministic finite automaton} (DFA) is a pair $(Q,\Sigma)$ of two finite sets equipped with a \emph{transition function} $Q\times\Sigma\to Q$ whose image at $(q,a)\in Q\times\Sigma$ is denoted by $q{\cdot}a$. We call $Q$ the \emph{state set} and $\Sigma$ the \emph{input alphabet}. Elements of $Q$ and $\Sigma$ are referred to as \emph{states} and, respectively, \emph{letters}, and for a state $q\in Q$ and a letter $a\in\Sigma$, we refer to $q{\cdot}a$ as the result of the \emph{action of $a$ at} $q\in Q$. A \emph{word over} $\Sigma$ is a finite sequence of letters. The action of letters in $\Sigma$ naturally extends to the action of words over $\Sigma$: If $w=a_1a_2\cdots a_n$ with $a_1,a_2,\dots,a_n\in\Sigma$, then $q{\cdot}w:=(\dots((q{\cdot}a_1){\cdot}a_2)\dots){\cdot}a_n$.

A DFA $(Q,\Sigma)$ is called \emph{synchronizing} if there exists a word $w$ over $\Sigma$ whose action brings the DFA to one particular state, regardless of the state at which $w$ is applied. This means that $q{\cdot}w=q'{\cdot}w$ for all $q,q'\in Q$. Any word $w$ with this property is said to be a \emph{reset} word for the automaton. Whether or not a DFA is synchronizing can be checked by analyzing the restriction of the power-set automaton to 2-element subsets, see the surveys \cite{San2005,Vol2008}.

Synchronizing automata serve as transparent and natural models of error-resis\-tant systems in many applications (coding theory, robotics, testing of reactive systems) and reveal interesting connections with symbolic dynamics, substitution systems, and other parts of mathematics. We refer the reader to chapter~\cite{KarVol2021} of the `Handbook of Automata Theory' and the survey \cite{Vol2022} for an introduction to the area and an overview of its state-of-the-art.

The fourth-named author initiated viewing \sa{} through the lens of game theory; the motivation for this came from a game-theoretical approach to software testing suggested in~\cite{BlaGNV2005}. We give a rigorous definition of games we consider in Section~\ref{sec:rules and presentation} as the following informal description suffices to explain the goals of the present paper.

In a synchronization game on a~DFA $\mathrsfs{A}$, two players, Alice (Synchronizer) and Bob (Desynchronizer), take turns choosing letters from the input alphabet of $\mA$. Alice, who wants to synchronize $\mathrsfs{A}$, wins when the sequence of chosen letters forms a reset word. Bob aims to prevent synchronization or, if it is unavoidable, delay it as long as possible. A bit more generally, one could have a subset~$S$ of states as an additional input and then Alice's task is to synchronize~$S$ while Bob has to prevent this.

Provided that both Alice and Bob play optimally, the outcome of such a game depends only on the automaton. This raises the problem of classifying \sa{} into those on which Alice and, respectively, Bob have a winning strategy. DFAs on which Alice can ensure a win are of interest because they are more amenable to synchronization, in a sense --- one may argue that such automata model devices that can be reliably controlled even under natural or artificial interference. For brevity, we call such DFAs \emph{A-automata}.

A few initial results on synchronization games were obtained in~\cite{FomVol2012,FomMarVol2013}. In particular, \cite[Theorem 4]{FomMarVol2013} provides an algorithm that, given a DFA $\mathrsfs{A}$ with $n$ states and $k$ input letters, decides who has a winning strategy in the synchronization game on~$\mA$ in $O(n^2k)$ time. The algorithm is based on the following reduction to 2-element subsets, similar to that used for recognizing synchronizability.

\begin{lemma}[{\mdseries (\cite[Lemma 2]{FomMarVol2013})}]\label{lemma:alice_strategy}
\label{lem:pair_condition_game}
	Alice has a winning strategy in the synchronization game on a DFA if and only if she has a winning strategy for the task of synchronizing every 2-element subset.
\end{lemma}
Thus, for any individual DFA, one can determine whether it is an A-automaton. Here, however, we are interested in general conditions ensuring that all synchronizing DFAs of a certain type are A-automata. One such condition was mentioned in~\cite{FomMarVol2013}: Alice always wins on definite DFAs introduced in~\cite{PerRabSha63}.

In the conference paper \cite{FerHaaHof2022}, the first three of the present authors showed that within two further families of automata considered in the literature---weakly acyclic DFAs and commutative DFAs---every \sDFA\ is an A-automaton. Our conference paper \cite{FerHHV2024} continued this line of research by designing a winning strategy for Alice that applies to \sa\ from yet another family of DFAs. Automata in this family are defined via a structure property of their transition monoids: the regular $\gD$-classes in these monoids form subsemigroups. The set $\mathbf{DS}$ of all finite monoids with this property plays a distinguished role in the algebraic theory of regular languages; see \cite[Chapter 8]{Alm95}. Therefore, DFAs with transition monoids in $\mathbf{DS}$ often appear in the literature; see, e.g., \cite{AlmMSV2009,AlmSte2009}. In~\cite{FerHHV2024} they were coined DS-\emph{automata} and the main result of~\cite{FerHHV2024} said that Alice can win the synchronization game on every synchronizing DS-automaton. Since the family of DS-automata encompasses all the families above (definite, weakly acyclic, and commutative DFAs), this vastly generalized the mentioned results from \cite{FomMarVol2013,FerHaaHof2022}.

The present paper is an enhanced version of the conference papers \cite{FerHaaHof2022,FerHHV2024}. To make it readable without prior acquaintance with \cite{FerHaaHof2022,FerHHV2024}, we have included in Sections~\ref{sec:rules and presentation} and~\ref{sec:prelim} a motivating section derived from~\cite{FerHaaHof2022} as well as necessary prerequisites taken from~\cite{FerHHV2024}. Section~\ref{sec:main} presents novel results. We propose a new winning strategy for Alice in synchronization games on synchronizing DS-automata. This new strategy is uniform: Alice's winning sequence of moves depends only on the underlying automaton and does not depend on Bob's responses, whereas the strategy in \cite{FerHHV2024} was adaptive. Then we show that among monoid pseudovarieties, $\mathbf{DS}$ is the largest class with the property that Alice has a uniform winning strategy in every synchronization game played on a \san{} whose transition monoid lies in this class. In Section~\ref{sec:discussion}, we discuss open questions and a generalization of our results to synchronization games in which the rule is modified in Bob's favor.

\section{Game rules and presentation}
\label{sec:rules and presentation}
\subsection{Game rules}
\label{subsec:rules}

We start with a visual yet rigorous description of the synchronization game under consideration. In this game, two players, Alice and Bob, play on a fixed DFA $\mathrsfs{A}=(Q,\Sigma)$. At the start, each state in $Q$ holds a token. During the game, some tokens can be removed according to the rules specified in the next paragraph. Alice wins if only one token remains, while Bob wins if he can keep at least two tokens unremoved for an indefinite amount of time.

Alice moves first, and then players alternate moves. The player whose turn is to move proceeds by selecting a letter $a\in\Sigma$. Then, for each state $q\in Q$ that held a token before the move, the token advances to the state $q{\cdot}a$. (In the standard graphical representation of $\mathrsfs{A}$ as a labelled digraph with $Q$ as its vertex set and the edges of the form $q\xrightarrow{a}q{\cdot}a$, one can visualize the move as follows: all tokens simultaneously slide along the edges labelled $a$.) If several tokens arrive at the same state after this, all of them but one are removed so that when the move is completed, each state holds at most one token.

To illustrate, let us look at how the synchronization game might play out on the DFA below in which, initially, each of the states 0,1,2 holds a token (shown in gray).
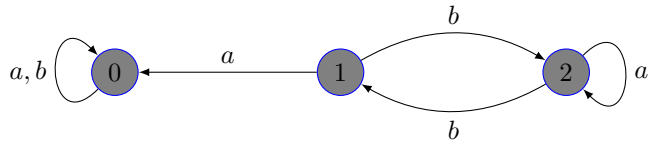
\begin{center}
\begin{tikzpicture}
	\node[fill=gray, circle, draw=blue, scale=1] (0) {$0$};
	\node[fill=gray, circle, draw=blue, scale=1, right of = 0, xshift= 2cm] (1) {$1$};
    \node[fill=gray, circle, draw=blue, scale=1, right of = 1, xshift= 2cm] (2) {$2$};
	\draw
		(1) edge[-latex, above]  node{$a$} (0)
        (1) edge[-latex, bend left, above]  node{$b$} (2)
		(2) edge[-latex, bend left, below] node{$b$}(1)
		(0) edge[-latex, loop left, left, in =135, out = -135, distance = 30] node{$a,b$} (0)
		(2) edge[-latex, loop right, right, out=45, in = -45, distance = 30] node{$a$} (2);
\end{tikzpicture}
\end{center}
If Alice chooses the letter $a$ on her first move, the tokens on states 0 and 2 remain due to the loops at these states. The token from state 1 moves to 0 and then is removed because the state 0 is `occupied'. Hence, the position after Alice's first move looks as follows:
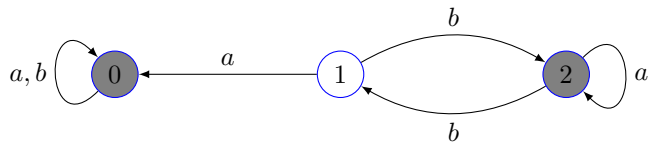
\begin{center}
\begin{tikzpicture}
	\node[fill=gray, circle, draw=blue, scale=1] (0) {$0$};
	\node[fill=white, circle, draw=blue, scale=1, right of = 0, xshift= 2cm] (1) {$1$};
    \node[fill=gray, circle, draw=blue, scale=1, right of = 1, xshift= 2cm] (2) {$2$};
	\draw
		(1) edge[-latex, above]  node{$a$} (0)
        (1) edge[-latex, bend left, above]  node{$b$} (2)
		(2) edge[-latex, bend left, below] node{$b$}(1)
		(0) edge[-latex, loop left, left, in =135, out = -135, distance = 30] node{$a,b$} (0)
		(2) edge[-latex, loop right, right, out=45, in = -45, distance = 30] node{$a$} (2);
\end{tikzpicture}
\end{center}
If Bob responds by choosing the letter $b$, the token on state 0 remains, while the token from state 2 moves to 1. Here is the position after Bob's response:
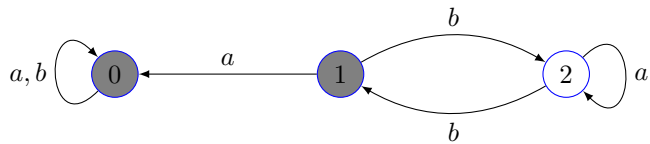
\begin{center}
\begin{tikzpicture}
	\node[fill=gray, circle, draw=blue, scale=1] (0) {$0$};
	\node[fill=gray, circle, draw=blue, scale=1, right of = 0, xshift= 2cm] (1) {$1$};
    \node[fill=white, circle, draw=blue, scale=1, right of = 1, xshift= 2cm] (2) {$2$};
	\draw
		(1) edge[-latex, above]  node{$a$} (0)
        (1) edge[-latex, bend left, above]  node{$b$} (2)
		(2) edge[-latex, bend left, below] node{$b$}(1)
		(0) edge[-latex, loop left, left, in =135, out = -135, distance = 30] node{$a,b$} (0)
		(2) edge[-latex, loop right, right, out=45, in = -45, distance = 30] node{$a$} (2);
\end{tikzpicture}
\end{center}
Now choosing $a$, Alice wins because after the token from state 1 moves to 0, it is removed, and we get the following position with only one token:
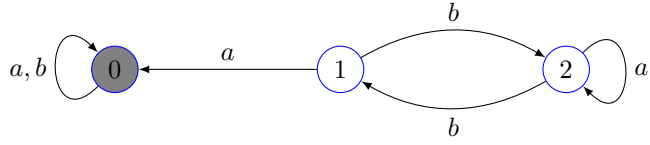
\begin{center}
\begin{tikzpicture}
	\node[fill=gray, circle, draw=blue, scale=1] (0) {$0$};
	\node[fill=white, circle, draw=blue, scale=1, right of = 0, xshift= 2cm] (1) {$1$};
    \node[fill=white, circle, draw=blue, scale=1, right of = 1, xshift= 2cm] (2) {$2$};
	\draw
		(1) edge[-latex, above]  node{$a$} (0)
        (1) edge[-latex, bend left, above]  node{$b$} (2)
		(2) edge[-latex, bend left, below] node{$b$}(1)
		(0) edge[-latex, loop left, left, in =135, out = -135, distance = 30] node{$a,b$} (0)
		(2) edge[-latex, loop right, right, out=45, in = -45, distance = 30] node{$a$} (2);
\end{tikzpicture}
\end{center}

Notice that Alice won the game described above only because of Bob's unfortunate response. In fact, Bob has a winning strategy in the synchronization game on this DFA: if he repeats Alice's moves, that is, chooses the same letter Alice chose on her previous move, he can maintain two tokens unremoved forever. Hence, there are simple synchronizing DFAs that are not A-automata.

\subsection{On board designs}
\label{subsec:board}

How can we present a synchronization game in a compact and attractive fashion? This is the question we like to discuss in this subsection. Clearly, every finite automaton can be presented by a transition table or by its automaton graph. However, we claim that more traditional designs like those known from traditional board games like chess or like games with racing tracks that are also popular children's games may be more appealing than the ones usually applied in automata theory. Notice that we cannot represent every automaton this way, but also notice that the proposed board designs do relate to the classes of automata studied throughout this paper.

\begin{figure}[htb]
\begin{center}
\begin{tikzpicture}[scale=0.72]
\draw [step =1, black,node distance=2cm]grid (7,5); \ua{0}{2.1}\da{1}{2}\la{3}{2}\la{4}{2}\ra{3}{1}\la{4}{1};
\draw [line width=0.7mm, black] (0,0) -- (0,5);
\draw [line width=0.7mm, black] (7,1) -- (7,5);
\draw [line width=0.7mm, black] (1,2) -- (1,3);
\draw [line width=0.7mm, black] (3,2) -- (5,2);
\draw [line width=0.7mm, black] (0,0) -- (6,0);
\draw [line width=0.7mm, black] (0,5) -- (7,5);

\filldraw[black] (0.5,0.45) circle (2pt) node[anchor=south] (p) {$p$};
\filldraw[black] (1.5,0.45) circle (2pt) node[anchor=south] (q) {$q$};
\filldraw[black] (0.5,2.5) circle (2pt) node[anchor=north] (r) {$r$};
\node at (4,-0.2,1) {\small (1) A board design with $p{\cdot}e=q$, $q{\cdot}w=p$.};
\end{tikzpicture}
\hspace{.8cm}
\begin{tikzpicture}[scale=0.72]
	\ua[red]{0}{2.2}\ua[red]{0}{1.9}
	\ra[red]{1.8}{4}\ra[red]{2}{4}\ra[red]{2.2}{4}
	\da[red]{6}{3}
	\da[red]{6}{1.2}\da[red]{6}{0.9}
	\la[red]{3.8}{0}\la[red]{4}{0}\la[red]{4.2}{0}
	
	\draw [line width=0.7mm, black] (0,0) -- (0,5);
	\draw [line width=0.7mm, black] (1,1) -- (1,4);
	\draw [line width=0.7mm, black] (7,0) -- (7,5);
	\draw [line width=0.7mm, black] (6,1) -- (6,4);
	\foreach \x in {1,...,4}
	{\draw [black] (0,\x) -- (1,\x);
		\draw [black] (6,\x) -- (7,\x);}
	\draw [line width=0.7mm, black] (0,0) -- (7,0);
	\draw [line width=0.7mm, black] (1,1) -- (6,1);
	\draw [line width=0.7mm, black] (0,5) -- (7,5);
	\draw [line width=0.7mm, black] (1,4) -- (6,4);
	
	\foreach \x in {1,...,6}
	{\draw [black] (\x,0) -- (\x,1);
		\draw [black] (\x,5) -- (\x,4);}
	
	\filldraw[black] (0.5,0.45) circle (2pt) node[anchor=south] (p) {$p$};
	\filldraw[black] (1.5,0.45) circle (2pt) node[anchor=south] (q) {$q$};
	\node at (4,-0.2,1) {\small (2) A racing track design};
\end{tikzpicture}
\end{center}
\caption{Designing boards for the synchronization game: two different proposals}
\label{fig:boards}
\end{figure}

Referring to Example (1) in Figure~\ref{fig:boards}, we can interpret such a board design as a finite automaton over the input alphabet $\{e,n,s,w\}$ (with the letters denoting the east, north, south and west directions of movement) as follows:
\begin{itemize}
 \item Each cell represents a state of the automaton.
\item The input letters let us move through the board as expected: $e$ moves one step ``east'', i.e., to the right, etc. For instance, if the leftmost lower cell represents the state $p$ and the cell to its right (i.e., the second cell on the lowest row of cells) represents state $q$, then $p{\cdot}e=q$ and $q{\cdot}w=p$ hold.
\item If the move would hit a wall, indicated by a thicker drawn line, then the direction indicated by the arrow in the current state should be ``executed.'' Therefore, in our picture, $r{\cdot}wsss=p=r{\cdot}esss=r{\cdot} nsss$.
\item If a wall is hit but no arrow is drawn, then the direction is inverted, i.e., if one bumps against the wall by moving east, a west move is executed, as $p{\cdot}w=q$ and $p{\cdot}sn=r$.\footnote{An alternative might be to get stuck when bumping against a wall; this would correspond to considering incomplete deterministic automata. This leads to different synchronization problems, and we like to avoid discussions in this direction here.}
\item An exception is the cell in the right lower corner: here the open walls indicate that this is a ``way out.'' More formally, there is one more state to be reached this way
(by moving either south or east), a sink state $\sigma$, which must be synchronizing. Hence, $q{\cdot}eeeeee=\sigma$.
\item Notice that apart from the exceptions that have been worked out above, the automaton behaves as a commutative automaton, e.g., $q{\cdot}nw=q{\cdot}wn$. This is also partially true for our rules at the walls, e.g., $p{\cdot} se=p{\cdot}es=q{\cdot}n$. However, $p{\cdot}sn=p{\cdot}nn\neq p{\cdot}ns=p$. But without the walls and in particular the ``special walls'' built in the middle of the board, the game might be a bit boring.
\end{itemize}

The goal of the game is as in the general case and can also be played with just two pieces (or tokens) on the board. The two players move both of them at the same time according to one of the letters $a,e,n,w$. While Player 1 will try to move both tokens on the same cell of the board, Player 2 will try to prevent this from happening.

Such a board design is quite compact. The example above denotes a 36-state DFA (35 cells in the grid plus the sink state $\sigma$), whose formal transition function is clearly less intuitive than this board presentation. Also when comparing it to a classical automaton graph representation, the suggested presentation has an edge, mainly because of the implicit transition arcs and implicit arc labels. Additionally, if one tries to draw an automaton graph as small as the board can be drawn without losing readability, the letters that need to be written on the arcs will be hard to decipher.

One can also use such a board to define a single-player game as follows. First, specify three positions $p,q,r$ on the board (an example is shown in Figure~\ref{fig:boards} (1), but in general you can think of special tasks that a player draws from a pile of tasks at the beginning of the game), plus a number~$n$. The question to solve is to find a sequence of movements, i.e., a word~$x$ of length at most~$n$ over the alphabet $\{e,n,s,w\}$, such that for the transition function $\delta$ that can be associated to the board, we find that $\delta(p,x)=\delta(q,x)=r$. With the positions shown in Figure~\ref{fig:boards} (1), it is not possible to find any such word. This is due to the fact that one can show that, assuming that the two pieces stay on the board, the Manhattan distance between the two pieces on the board (ignoring walls) will always be an odd number.  In general, this specific form of a synchronizability question could be interesting as a single-player game, although it is polynomial-time solvable (without the length restriction), mainly because of the following facts: (a) as discussed above, boards can be used as quite compact representations of DFA; (b) although words that synchronize two states are of quadratic size only, this might mean that the player might have to look for a synchronizing word of a length like 100 even for the small board displayed in Figure~\ref{fig:boards}. This makes this question quite challenging for a human player.

A second example is illustrated in Figure \ref{fig:boards} (2). Here, the underlying automaton consists of 20 states, each representing a square, and two input letters, $a$ and $b$ whose actions are defined as follows. If a square contains a token, applying $b$ moves the token to the next square clockwise. If a square with a token has $k$ red arrows, applying $a$ moves the token $k$ steps in the clockwise direction; if no red arrow is present,  $a$ leaves the token at the same square. In the initial setup, two tokens ($p$ and $q$ in Figure \ref{fig:boards} (2)) are placed on the board. In each round, players choose to move both tokens by applying either $a$ or $b$. Player 1 aims to synchronize the tokens (i.e., land them on the same square). Player 2  tries to prevent synchronization. This `racing track' game is, obviously, equivalent to the synchronization game by Lemma~\ref{lemma:alice_strategy}.

Apart from being more similar to board games, this perspective on finite automata is also motivated by robot navigation problems; see \cite{RivSch93}. Conceptually, it is interesting to note that this particular paper talks about homing sequences, a notion quite akin to that of synchronizing words, cf. the exposition of Sandberg~\cite{San2005}.  Of course, one could think of different puzzles played on such a board: the traditional \textsc{Short Synchro} problem would lead to a one-person game, while the synchronization game itself is a two-person game. Further variations (or possibly levels) can be designed by introducing further special ``effects'' (or board cell symbols). For instance, one could introduce cells where the directions are rather interpreted as knight moves, etc.

Let us note that one could also use non-standard board designs, like hexagonal grids, or even graphs themselves. Also note that the first proposed board design takes some ideas from chain code picture languages, another area from which one could develop educational games with tight links to automata theory, see \cite{AbeSes82,Bri2000,CulDub93,CulDub93d,DasHin93,Gut89,Nol94,SudWel85}. An alternative would be to consider turtle-like commands to lead through a maze, interpreted as a finite automaton.

Finally, we would like to suggest to use these types of board designs when introducing finite automata in first year's Computer Science courses or also in schools, making this material more accessible and attractive. Then, one could also explain one of the major skills computer scientists should acquire: that of formally modelling environments. For instance, coming back to a horse track example, akin to a children's dice game: input letters could then be $\{1,\dots,6\}$, and two competing horses could be modeled by a product automaton construction (which would also nicely illustrate this formalism); one could then discuss how to introduce or model notions like suspending etc. In the context of this paper, it is interesting to finally note that the concept of an `environment' modeled as a finite automaton is quite common in the literature of Artificial Intelligence, often not mentioning the concept of finite automata explicitly; we cite only \cite{Nat89}, \cite{Suomalainen:2020} and \cite{ArrighiFO023} where the notion of reset (or synchronizing) words does show up.

\section{Algebraic preliminaries}
\label{sec:prelim}

Our approach to synchronization games is algebraic, as it exploits certain structural properties of transition monoids. These are collected in this section.

\subsection{Transition monoids and synchronization}
\label{subsec:transmonoid}

Let $\mathrsfs{A}=(Q,\Sigma)$ be a DFA. For each letter $a\in\Sigma$, the map $\tau_a\colon Q\to Q$ defined by the rule $q\mapsto q{\cdot}a$ is a transformation on the set $Q$.
\begin{definition}
\label{def:transmonoid}
The \emph{transition monoid} of a DFA $\mA=(Q,\Sigma)$ is the submonoid of the monoid of all transformations on the set $Q$ generated by the set $\{\tau_a\mid a\in\Sigma\}$.
\end{definition}

We denote the transition monoid of a DFA $\mA=(Q,\Sigma)$ by $T(\mA)$. It is easy to see that any product $\tau_{a_1}\tau_{a_2}\cdots\tau_{a_n}$ with $a_1,a_2,\dots,a_n\in\Sigma$ is precisely the transformation $\tau_w$ defined by the rule $q\mapsto q{\cdot}w$ where $w$ stands for the word $a_1a_2\cdots a_n$. Thus, the transition monoid $T(\mA)$ can alternatively be defined as the monoid of all transformations on the set $Q$ caused by the action of words over $\Sigma$, provided one adopts the standard convention that the identity transformation is caused by the empty word.

If $\mathrsfs{A}=(Q,\Sigma)$ is a synchronizing DFA and $w$ is a reset word for $\mA$, then the transformation $\tau_w$ is a constant map on $Q$, that is, $Q\tau_w=\{q\}$ for a certain $q\in Q$. Thus, the transition monoid of a \san{} always contains a constant transformation. Conversely, if $\zeta\in T(\mA)$ is a constant transformation, then any word $w$ with $\tau_w=\zeta$ is a reset word for $\mA$ and so $\mA$ is synchronizing. We see that synchronization is actually a property of the transition monoid of an automaton rather than the automaton itself: for DFAs $\mathrsfs{A}=(Q,\Sigma)$ and $\mathrsfs{A}'=(Q,\Sigma')$ with the same state set but different input alphabets, the equality $T(\mA)=T(\mA')$ guarantees that $\mA'$ is synchronizing if and only if so is $\mA$.

We can say a bit more about transition monoids of \sa{}, but for this, we first need to recall some concepts of semigroup theory. A non\-empty subset~$I$ of a semigroup $S$ is called an \emph{ideal} in~$S$ if, for all $s\in S$ and $i\in I$, both products $si$ and $is$ lie in $I$. If $I$ and $J$ are two ideals in $S$, their intersection $I\cap J$ contains any product $ij$ with $i\in I$, $j\in J$. So $I\cap J$ is nonempty, and then it is easy to verify that  $I\cap J$ is again an ideal in $S$. Therefore, each finite semigroup $S$ has a least ideal being the intersection of all ideals of $S$. This least ideal is called the \emph{kernel} of $S$ and denoted by $\Ker S$.

The following observation is folklore but we provide a proof for completeness.
\begin{lemma}
\label{lem:kernel}
A DFA is synchronizing if and only if the kernel of its transition monoid consists of constant transformations.
\end{lemma}

\begin{proof}
The `if' part readily follows from the already mentioned fact that if the transition monoid of a DFA contains a constant transformation, then the DFA is synchronizing.

For the `only if' part, let $\mA$ be a \sDFA; then the set $C$ of all constant transformations in the transition monoid $T(\mA)$ is nonempty. For any $\tau\in T(\mA)$ and any $\zeta\in C$, we have
\begin{equation}\label{eq:rightzero}
\tau\zeta=\zeta.
\end{equation}
Equality~\eqref{eq:rightzero} implies that the set $C$ is contained in every ideal of the monoid $T(\mA)$, in particular, in its kernel $\Ker T(\mA)$. Further, for every $\tau\in T(\mA)$ and every $\zeta\in C$, the product $\zeta\tau$ is a constant transformation: if $Q\zeta=\{q\}$, then $Q\zeta\tau=\{q\tau\}$. Together with \eqref{eq:rightzero}, this observation implies that $C$  forms an ideal in $T(\mA)$. Since $\Ker T(\mA)$ contains $C$, we have $\Ker T(\mA)=C$ by the definition of the kernel.
\end{proof}

\subsection{Structure of monoids in\/ \textup{\textbf{DS}}}

Green \cite{Gre51} defined five important relations on every semigroup $S$, which are collectively referred to as \emph{Green's relations}. Of these, we shall need the following three:
\begin{itemize}
\item[] $a\gR b \Longleftrightarrow {}$ either $a=b$ or $a=bs$ and $b=at$ for some $s,t\in S$;
\item[] $a\gL \,b \Longleftrightarrow {}$ either $a=b$ or $a=sb$ and $b=ta$ for some $s,t\in S$;
\item[] $a\gD b \Longleftrightarrow {}$ $a\gR c$ and $c\gL b$ for some $c\in S$.
\end{itemize}
The relations $\gR$ and $\gL$ are obviously equivalencies. The definition of ${\gD}$ means that ${\gD}$ is the product of $\gR$ and $\gL$ as binary relations. As observed in \cite{Gre51}, ${\gD}$ is also the product of $\gL$ and $\gR$, whence $\gD$ is the least equivalence containing both $\gR$ and $\gL$.

An element $a$ of a semigroup $S$ is said to be \emph{regular} if $asa=a$ for some $s\in S$. A $\gD$-class $D$ is called \emph{regular} if it contains a regular element. (In this case, every element of $D$ is known to be regular; see \cite[Theorem 6]{Gre51}.) We denote by $\mathbf{DS}$ the set of all finite monoids $S$ such that each regular $\gD$-class of $S$ is a subsemigroup in $S$.

The structure of monoids in $\mathbf{DS}$ is well understood in terms of their decomposition into some basic blocks. As we use this structural result, we recall the notions involved.

A \emph{semilattice} is a semigroup that satisfies the laws of commutativity $xy=yx$ and idempotency $x^2=x$ for all elements $x,y$.

\begin{definition}
\label{def:semilattice}
Let $Y$ be a semilattice and $\{S_y\}_{y\in Y}$ a family of disjoint semigroups indexed by the elements of $Y$. A semigroup $S$ is said to be a \emph{semilattice of semigroups} $S_y$, with $y\in Y$, if:
\begin{description}
  \item[{\normalfont(S1)}] $S=\bigcup_{y\in Y}S_y$;
  \item[{\normalfont(S2)}] each $S_y$ is a subsemigroup in $S$;
  \item[{\normalfont(S3)}] for every $y,z\in Y$ and every $s\in S_y$, $t\in S_z$, the product $st$ belongs to $S_{yz}$.
\end{description}
\end{definition}

A semigroup $S$ is $m$-\emph{nilpotent over its kernel} ($m$ being a positive integer) if every product of $m$ elements of $S$ belongs to $\Ker S$. We call a semigroup \emph{nilpotent over its kernel} if it is $m$-nilpotent over its kernel for some $m$. (To a semigroupist, finite semigroups nilpotent over their kernels are familiar as finite \emph{Archimedean} semigroups.)

The following is a specialization of the equivalence (4c) $\Leftrightarrow$ (1b) in \cite[Theorem 3]{She95a} to finite monoids\footnote{Theorem 3 in \cite{She95a} deals with semigroups in which every element has a power that belongs to a subgroup. Every finite semigroup has this property.}.
\begin{lemma}
\label{lem:ds}
Every monoid in\/ $\mathbf{DS}$ is a semilattice of semigroups that are nilpotent over their kernels.
\end{lemma}

\subsection{A characterization of\/  \textup{\textbf{DS}} as a pseudovariety}

A monoid \emph{pseudovariety} is a set of finite monoids closed under taking homomorphic images, submonoids and finite direct products. Monoid pseudovarieties are actively studied because of their tight connections to certain classes of recognizable languages via the Eilenberg correspondence (\cite{Eilenberg1976}; see also~\cite{Pin1984}). The set $\mathbf{DS}$ is a pseudovariety; it and several of its subpseudovarieties play an important role in  studies on the Eilenberg correspondence; see \cite[Chapter 8]{Alm95}.

The 6-element \emph{Brandt monoid} $B_2^1$ consists of the following six $2\times 2$-matrices, multiplied according to the usual matrix multiplication rule:
\[
\begin{pmatrix}
1 & 0 \\ 0 & 1
\end{pmatrix}, ~~
\begin{pmatrix}
1 & 0 \\ 0 & 0
\end{pmatrix}, ~~
\begin{pmatrix}
0 & 1 \\ 0 & 0
\end{pmatrix}, ~~
\begin{pmatrix}
0 & 0 \\ 1 & 0
\end{pmatrix}, ~~
\begin{pmatrix}
0 & 0 \\ 0 & 1
\end{pmatrix}, ~~
\begin{pmatrix}
0 & 0 \\ 0 & 0
\end{pmatrix}.
\]
It is known (and easy to verify) that the monoid $B_2^1$ is isomorphic to the transition monoid of the automaton $\mB_2$ shown in Figure~\ref{fig:autB21}.
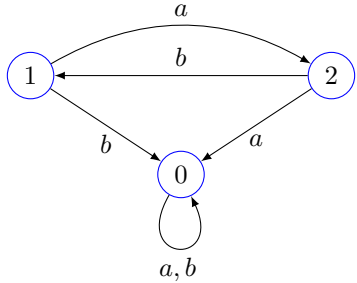
\begin{figure}[htb]
\begin{center}
\begin{tikzpicture}
	\node[fill=white, circle, draw=blue, scale=1] (0) {$0$};
	\node[fill=white, circle, draw=blue, scale=1, above = 0, xshift= -2cm, yshift= 1cm] (1) {$1$};
    \node[fill=white, circle, draw=blue, scale=1, above = 0, xshift= 2cm, yshift= 1cm] (2) {$2$};
	\draw
		(1) edge[-latex, below]  node{$b$} (0)
        (1) edge[-latex, bend left, above]  node{$a$} (2)
		(2) edge[-latex, above] node{$b$}(1)
		(2) edge[-latex, below] node{$a$} (0)
		(0) edge[-latex, loop below, below, in =-65, out = -120, distance = 30] node{$a,b$} (0);
		\end{tikzpicture}
\end{center}
\caption{The automaton $\mB_2$ with the transition monoid isomorphic to $B_2^1$}
\label{fig:autB21}
\end{figure}

We need the following characterization of $\mathbf{DS}$.

\begin{lemma}[{\mdseries\cite[Theorem 3]{Mar81}}]\label{lem:exclds}
$\mathbf{DS}$ is the largest monoid pseudovariety that does not contain the monoid $B_2^1$.
\end{lemma}

\section{Uniform winning strategy in synchronization games on DS-automata}
\label{sec:main}

\subsection{Uniform winning strategies: definition and examples}
\label{subsec:uniform}

A \emph{uniform winning strategy} for a DFA $\mA$ is a word $w$ over the input alphabet of $\mA$ with the property that in the synchronization game on $\mA$, if Alice plays the letters of $w$ in order (using the $i$-th letter of $w$ as her $i$-th move), then she wins the game, regardless of Bob's responses. We denote the class of \sa{} with uniform winning strategies by \UWS.

Synchronizing automata of several previously studied types lie in \UWS. As an example, consider definite automata. This DFA family was introduced by some of the pioneers of automata theory back in 1963~\cite{PerRabSha63}. In~\cite{PerRabSha63}, the term `automaton' meant a recognizer, that is, a DFA with a designated initial state and a distinguished set of final states. However, DFAs without initial and final states as defined in the present paper also appeared in~\cite{PerRabSha63} but under the name `transition tables'. The following is \cite[Definition 13]{PerRabSha63} stated in our terminology and notation.
\begin{definition}
\label{def:definite}
A DFA $(Q,\Sigma)$ is \emph{weakly $k$-definite} if $q{\cdot}w=q'{\cdot}w$ for every word $w$ of length at least~$k$ over $\Sigma$ and all $q,q'\in Q$. A DFA is $k$-\emph{definite} if it is weakly $k$-definite but not weakly $(k-1)$-definite and is \emph{definite} if it is $k$-definite for some $k$.
\end{definition}

By Definition~\ref{def:definite}, definite DFAs are synchronizing and for each $k$-definite DFA $\mA=(Q,\Sigma)$, every word of length at least $k$ over $\Sigma$ is a reset word.  Therefore, every word of length at least $k/2$ over $\Sigma$ is a uniform winning strategy for $\mA$.

A DFA is called \emph{commutative} if its transition monoid is commutative, that is, it satisfies the law $xy=yx$. Synchronizing commutative automata were considered in \cite{Rys96,Rys97,Hof2021m}. A simple winning strategy for Alice in synchronization games on such automata was suggested in~\cite[Theorem 5.2]{FerHaaHof2022}: Alice must just choose letters spelling out a reset word, ignoring the moves of Bob. Thus, any reset word for a commutative \san{} $\mA$ is a uniform winning strategy for $\mA$.

The automaton $\mB_2$ shown in Figure~\ref{fig:autB21} yields another example. One can verify that each of the words $ab$ and $ba$ is a uniform winning strategy for $\mB_2$.

\subsection{Uniform winning strategy within\/ \textup{\textbf{DS}}}

Recall that a DS-automaton is a DFA with transition monoid in the pseudovariety $\mathbf{DS}$. The following strengthening of \cite[Theorem 3]{FerHHV2024} is the main result of this paper.

\begin{theorem}
\label{thm:main}
Every synchronizing DS-automaton admits a uniform winning strategy.
\end{theorem}

\begin{proof}
Take an arbitrary synchronizing DS-automaton $\mA=(Q,\Sigma)$. We denote the transition monoid $T(\mA)$ by $S$ to lighten the notation. Since $S\in\mathbf{DS}$, by Lemma~\ref{lem:ds} there is a semilattice $Y$ such that $S$ is a semilattice of semigroups $S_y$, $y\in Y$, where each semigroup $S_y$ is nilpotent over its kernel.

On every semilattice, and thus on $Y$, the relation $\le$ defined by
\[
x\le y\Longleftrightarrow  xy=x
\]
is known (and easy to see) to be a partial order under which, for any $y_1,\dots,y_k\in Y$, the product $y_1\cdots y_k$ is the greatest lower bound of the set $\{y_1,\dots,y_k\}$. Since the semilattice $Y$ is finite, the product $z$ of all its elements is the least element with respect to $\le$. We have $yz=zy=z$ for every $y\in Y$; hence by item (S3) in Definition~\ref{def:semilattice}, the semigroup $S_z$ is an ideal in $S$. Therefore $S_z$ contains the kernel $\Ker S$ of $S$, and thus $\Ker S_z\subseteq\Ker S$. (In fact, it is easy to show that $\Ker S_z=\Ker S$, but this is not needed for the present proof.)

Fix a positive integer $m$ such that the semigroup $S_z$ is $m$-nilpotent over its kernel $\Ker S_z$ and fix an element $\zeta\in S_z$. Since the monoid $S$ consists of transformations caused by the action of words over $\Sigma$, there is a word $w$ such that $\zeta=\tau_w$. We claim that the word $w^m$ is a uniform winning strategy for $\mA$.

Recall that $S=\bigcup_{y\in Y}S_y$ by item (S1) in Definition~\ref{def:semilattice}, and that the semigroups $S_y$ are disjoint. Therefore, for each letter $a\in\Sigma$, there is a unique element $y(a)\in Y$ such that the transformation $\tau_a$ lies in the subsemigroup $S_{y(a)}$. Write the word $w$ as a product of letters: $w=a_1a_2\cdots a_n$. Suppose that Alice spells out the word $w$ on her moves and that Bob responds to Alice's $i$-th move, choosing a letter $b_i\in\Sigma$. Let
\[
u:=a_1b_1a_2b_2\cdots a_nb_n,
\]
and let $x\in Y$ be such that $\tau_u\in S_x$. We have
\begin{align*}
  x&=y(a_1)y(b_1)y(a_2)y(b_2)\cdots y(a_n)y(b_n) &&\text{by item (S3) in Definition~\ref{def:semilattice}} \\
   &\le y(a_1)y(a_2)\cdots y(a_n) &&\text{by the definition of $\le$}\\
   &=z && \text{since $\tau_{a_1}\tau_{a_2}\cdots\tau_{a_n}=\tau_w=\zeta\in S_z$}.
\end{align*}
Since $z$ is the least element with respect to $\le$, we conclude that $x=z$. Hence, the transformation $\tau_u$ belongs to $S_z$.

Now let Alice spell out the word $w^m$, and let $v$ be the word obtained from $w^m$ by inserting the letters chosen by Bob after each of Alice's moves. Clearly, $v$ can be decomposed as
\[
v=u_1\cdots u_m
\]
where $u_j$ is the word obtained from the $j$-th copy of $w$ after inserting Bob's responses corresponding to that copy. Applying the argument of the preceding paragraph to each $u_j$ shows that, for every $j=1,\dots,m$, the transformation $\tau_{u_j}$ belongs to $S_z$. Then the transformation $\tau_v$ is a product of $m$ elements of $S_z$ and so $\tau_w$ belongs to $\Ker S_z$ as the semigroup $S_z$ is $m$-nilpotent over its kernel. Since $\mA$ is a \san, the kernel $\Ker S$ of its transition monoid consists of constant transformations by Lemma~\ref{lem:kernel}. From the inclusion $\Ker S_z\subseteq\Ker S$ registered above, we conclude that $\tau_w$ is a constant transformation, and so $v$ is a reset word for $\mA$.

Therefore, by spelling out the word $w^m$, Alice wins regardless of Bob's responses. This establishes our claim that $w^m$ is a uniform winning strategy for $\mA$.
\end{proof}

\subsection{Corollaries}
\label{subsec:corollaries}

In Section~\ref{subsec:uniform},  we mentioned two previously known families of DFAs within \UWS: definite automata and commutative \sa. These families consist of DS-automata so their uniform winning strategies are subsumed by that of Theorem~\ref{thm:main}.

The transition monoid of a definite DFA is nilpotent over its kernel. This fact is implicitly contained in~\cite{PerRabSha63} and in the explicit form, it is a part of~\cite[Theorem~3]{Zal72}. Comparing it with Lemma~\ref{lem:ds}, we see that definite automata constitute a special subfamily of DS-automata.

Obviously, on any commutative semigroup, Green's relations $\gR$ and $\gL$ coincide with each other, and hence, with the relation $\gD$. If an element $a$ of a commutative semigroup $S$ is regular, then $a=a^2s$ for some $s\in S$ whence $a\gR a^2$. By \cite[Theorem 7]{Gre51}, this ensures that the $\gD$-class containing $a$ is a subsemigroup (and even a subgroup) of $S$. Thus, all finite commutative monoids lie in $\mathbf{DS}$, and hence, commutative DFAs are DS-automata.

Another family of DFAs considered in connection with synchronization games is that of weakly acyclic automata. A DFA $(Q,\Sigma)$ is \emph{weakly acyclic} if for any $p,q\in Q$, the existence of words $u,v$ over $\Sigma$ such that $q=p{\cdot}u$ and $p=q{\cdot}v$ implies $p=q$. In terms of the graphical representation of DFAs, weakly acyclic automata are characterized by the property that the only directed cycles in their digraphs are loops.

Various properties of synchronizing weakly acyclic DFAs were considered in \cite{Ryz2019a,Hof2021g}. A winning strategy for Alice in synchronization games on such automata was suggested in~\cite[Theorem 2.3]{FerHaaHof2022}.

By~\cite[Proposition 6.2]{BrzFic80} the transition monoid $T(\mA)$ of any weakly acyclic DFA $\mA$ is $\gR$-\emph{trivial}, meaning that Green's relation $\gR$ on $T(\mA)$ coincides with the equality relation. It is well known that all finite $\gR$-trivial monoids lie in $\mathbf{DS}$; see, e.g., \cite[Chapter 9]{Alm95}. Thus, weakly acyclic DFAs are DS-automata, and Theorem~\ref{thm:main} provides uniform winning strategies for weakly acyclic  \sa{}.

\subsection{No uniform winning strategy beyond\/ \textup{\textbf{DS}}}

Consider the automaton $\mB'_2$ obtained from the DFA $\mB_2$ shown in Figure~\ref{fig:autB21} by adding an extra letter $c$ that acts as the identity transformation; see Figure~\ref{fig:autB2prime}.
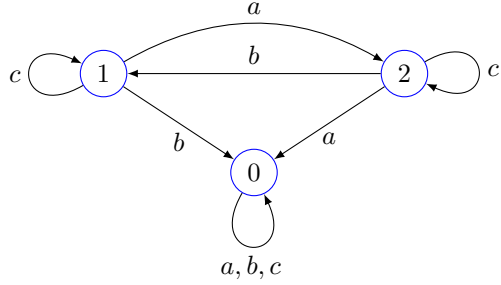
\begin{figure}[htb]
\begin{center}
\begin{tikzpicture}
	\node[fill=white, circle, draw=blue, scale=1] (0) {$0$};
	\node[fill=white, circle, draw=blue, scale=1, above = 0, xshift= -2cm, yshift= 1cm] (1) {$1$};
    \node[fill=white, circle, draw=blue, scale=1, above = 0, xshift= 2cm, yshift= 1cm] (2) {$2$};
	\draw
		(1) edge[-latex, below]  node{$b$} (0)
        (1) edge[-latex, bend left, above]  node{$a$} (2)
		(2) edge[-latex, above] node{$b$}(1)
		(2) edge[-latex, below] node{$a$} (0)
		(0) edge[-latex, loop below, below, in =-65, out = -120, distance = 30] node{$a,b,c$} (0)
        (1) edge[-latex, loop left, left, in =155, out = -150, distance = 30] node{$c$} (1)
        (2) edge[-latex, loop right, right, in =-25, out = 30, distance = 30] node{$c$} (2);
		\end{tikzpicture}
\end{center}
\caption{The automaton $\mB'_2$ }
\label{fig:autB2prime}
\end{figure}
Adding the identity transformation as a letter does not change the transition monoid of a DFA, as the monoid already contains the identity element by definition. Hence, the transition monoid $T(\mB'_2)$ is the same as $T(\mB_2)$, and so it is isomorphic to the monoid $B_2^1$.

\begin{lemma}\label{lem:B2prime}
Alice can always win the synchronization game on the automaton $\mB'_2$, but no uniform winning strategy for $\mB'_2$ exists.
\end{lemma}

\begin{proof}
Alice wins the synchronization game on $\mB'_2$ by choosing $a$ as her first move. If Bob responds with $a$, he already loses; if he responds with $b$ or $c$, then Alice wins by choosing $b$ in the former case and $a$ in the latter.

Now consider an arbitrary word $w=x_1x_2\cdots x_n$ over the input alphabet $\{a,b,c\}$ of $\mB'_2$. For $i=1,2,\dots,n$, insert a letter $y_i\in\{a,b,c\}$ after the letter $x_i$, using the following rules:
\begin{itemize}
  \item if $i=n$ or if $i<n$ and $x_i\ne x_{i+1}$ or $x_i=x_{i+1}=c$, then $y_i:=c$;
  \item $i<n$ and $x_i=x_{i+1}=a$, then $y_i:=b$;
  \item $i<n$ and $x_i=x_{i+1}=b$, then $y_i:=a$.
\end{itemize}
In the resulting word $w'=x_1y_1x_2y_2\cdots x_ny_n$, occurrences of $a$ and $b$ alternate, possibly with blocks of $c$'s inserted between them:
\[
w'=c^{m_1}ac^{m_2}bc^{m_3}ac^{m_4}b\cdots c\quad \text{ or } \quad c^{m_1}bc^{m_2}bc^{m_3}bc^{m_4}a\cdots c
\]
where each $m_i$ is a nonnegative integer (any power with exponent~0 is understood as the empty word).  Since the letter $c$ acts as the identity transformation, the action of $w'$ equals the action of the word $w''$ obtained from $w'$ by omitting all occurrences of $c$. Thus, $w''$ is an alternating word in $a$ and $b$, and has one of the forms $(ab)^ka$, $(ba)^kb$, $(ab)^k$, or $(ba)^k$, where $k$ is a nonnegative integer. None of these is a reset word for $\mB'_2$, whence $w$ is not a uniform winning strategy for $\mB'_2$.
\end{proof}
We can now deduce the following corollary.
\begin{corollary}\label{cor:optimal}
Let\/ $\mathbf{P}$ be a monoid pseudovariety. Uniform winning strategies exist for all \sDFAs{} with transition monoids in\/ $\mathbf{P}$ if and only if\/ $\mathbf{P}\subseteq\mathbf{DS}$.
\end{corollary}

\begin{proof}
The `if' part follows from Theorem~\ref{thm:main}.

For the `only if' part, combine Lemmas~\ref{lem:exclds} and~\ref{lem:B2prime} with the observation made before Lemma~\ref{lem:B2prime} that the transition monoid of the DFA $\mB'_2$ is isomorphic to the monoid $B_2^1$.
\end{proof}

\section{Open questions and a generalization}
\label{sec:discussion}

\subsection{Decidability and complexity}
\label{subsec:decide}

Our description of monoid pseudovarieties containing transition monoids of only those \sDFAs{} that admit uniform winning strategies does not yield a characterization of the class \UWS{} of all such automata. Moreover, unlike synchronization, the property of having a uniform winning strategy (and likewise the property of being an A-automaton) cannot be detected from transition monoids. We have already exhibited two DFAs, $\mB_2$ and $\mB'_2$ (see Figures~\ref{fig:autB21} and~\ref{fig:autB2prime}), that have identical transition monoids, even though one admits a uniform winning strategy and the other does not. In fact, for every \sDFA{} $\mathrsfs{A}=(Q,\Sigma)$, there exists an A-automaton $\mathrsfs{A}'=(Q,\Sigma')$  such that $T(\mA)=T(\mA')$. To see this, let $\Sigma':= \Sigma\cup\{c\}$, where the action of the added letter $c$ coincides with the action of some fixed reset word for $\mA$. The transformations induced by the letters in $\Sigma'$ generate the same submonoid of the full transformation monoid on $Q$ as do those induced by the letters in $\Sigma$. Still, Alice wins the synchronization game on $\mA'$ instantly by choosing the letter $c$ on her first move; in particular, the one-letter word $c$ is a uniform winning strategy for $\mA'$.

It is not hard to see that membership in the class \UWS{} is decidable. If Alice has a uniform winning strategy for a DFA $\mathrsfs{A}=(Q,\Sigma)$, let $w=a_1a_2\cdots a_k$ be such a strategy of minimum length. Define inductively a sequence $C_0,C_1,\dots,C_i,\dots$ of nonempty sets of nonempty subsets of $Q$, called \emph{configurations}, by setting $C_0:=\{Q\}$ and, for $i=1,2,\dotsc$,
\[
C_i:=\{P{\cdot}a_ix \mid P\in C_{i-1}, x\in\Sigma\}, \ \text{ where } \ P{\cdot}a_ix:=\{p{\cdot}a_ix\mid p\in P\}.
\]
Thus, the configuration $C_i$ consists of all subsets of states that may hold tokens after Alice's $i$-th move (namely $a_i$) followed by an arbitrary response of Bob. That $w$ is a uniform winning strategy means that the final configuration $C_k$ consists entirely of singleton sets. The minimality of $w$ ensures that no two configurations can be equal: indeed, if $C_i=C_j$ for some $i<j$, then the shorter word $a_1a_2\cdots a_ia_{j+1}\cdots a_k$ would also be a uniform winning strategy for $\mA$. Hence, the length of $w$ does not exceed the number of distinct configurations that contain at least one nonsingleton set, that is, the number $K:=2^{\,2^{|Q|} - 1} - 2^{\,|Q|}$.
Thus, one can decide whether there exists a uniform winning strategy for $\mathrsfs{A}=(Q,\Sigma)$ by checking, for each of the $|\Sigma|^{\frac{K(K+1)}{2}}$ many words $a_1a_2\cdots a_k$ of length at most $K$ over $\Sigma$, whether $\mathrsfs{A}$ is reset by every word in the finite set
\[
\{a_1b_1a_2b_2\cdots a_kb_k \mid b_1,b_2,\dots,b_k\in\Sigma\}.
\]

Of course, the described brute-force procedure is extremely inefficient, and it is natural to ask whether an efficient algorithm exists---say, one running in time (or at least space) polynomial in the number of states. At present, we do not know the answer to this question.

\subsection{Bounds on uniform winning strategy length}
\label{subsec:bounds}

The only upper bound on the length of a uniform winning strategy for an automaton with $n$ states in the class \UWS{} known to us is the double-exponential bound $2^{\,2^{n} - 1} - 2^{\,n}$ established in Section~\ref{subsec:decide}.

For comparison, Alice requires at most $\binom{n}{2}(n-2)+1$ moves to win the synchronization game on any A-automaton with $n$ states \cite[Corollary 3]{FomMarVol2013}. This cubic bound follows from a `localization' lemma (see Lemma~\ref{lem:pair_condition_game}), which reduces finding a winning strategy to the task of synchronizing 2-element subsets, together with the observation that Alice can achieve synchronization of each 2-element subset in an A-automaton with $n$ states in at most $\binom{n}{2}$ moves. We do not know whether either of these two arguments can be extended to uniform winning strategies; we can only prove that in any $n$-state automaton from \UWS, Alice can synchronize each 2-element subset, following a uniform strategy of length at most $2^{\binom{n+1}{2}}-2^n$.

Regarding lower bounds, we can extract a quadratic lower bound on the length of a uniform winning strategy for an automaton with $n$ states from a construction in the proof of \cite[Theorem 6]{FomMarVol2013}. We reproduce this construction for the reader's convenience.

Let $\mathrsfs{A}=(Q,\Sigma)$ be a DFA with $|\Sigma|\ge2$. Select a letter $b\in\Sigma$ and a state $q_0\in Q$, and let $\mathrsfs{D}=(Q\times\{0,1\},\Sigma)$, where for each $q\in Q$, the action of an arbitrary letter $a\in\Sigma$ is defined as follows:
\[(
q,0){\cdot}a=(q{\cdot}a,1), \qquad (q,1){\cdot}a=
  \begin{cases}
    (q,0) & \text{if $a=b$}, \\
    (q_0,1) & \text{otherwise}.
  \end{cases}
\]
The DFA $\mathrsfs{D}$ is called the \emph{duplication} of $\mathrsfs{A}$.

For illustration, Figure~\ref{fig:cerny-n} presents the duplication of the \v{C}ern\'{y} automaton $\mathrsfs{C}_n$ from \cite{Cer64} (with the state $0$ in the role of $q_0$ and $b$ as the selected letter). We recall that the states of $\mathrsfs{C}_n$ are the residues modulo $n$ and the input letters $a$ and $b$ act as follows:
\[
m{\cdot}a=
 \begin{cases}
  1 & \text{for $m = 0$}, \\
  m & \text{for $1\le m<n$};
  \end{cases}
\qquad m{\cdot}b=m+1\!\!\pmod{n}.
\]

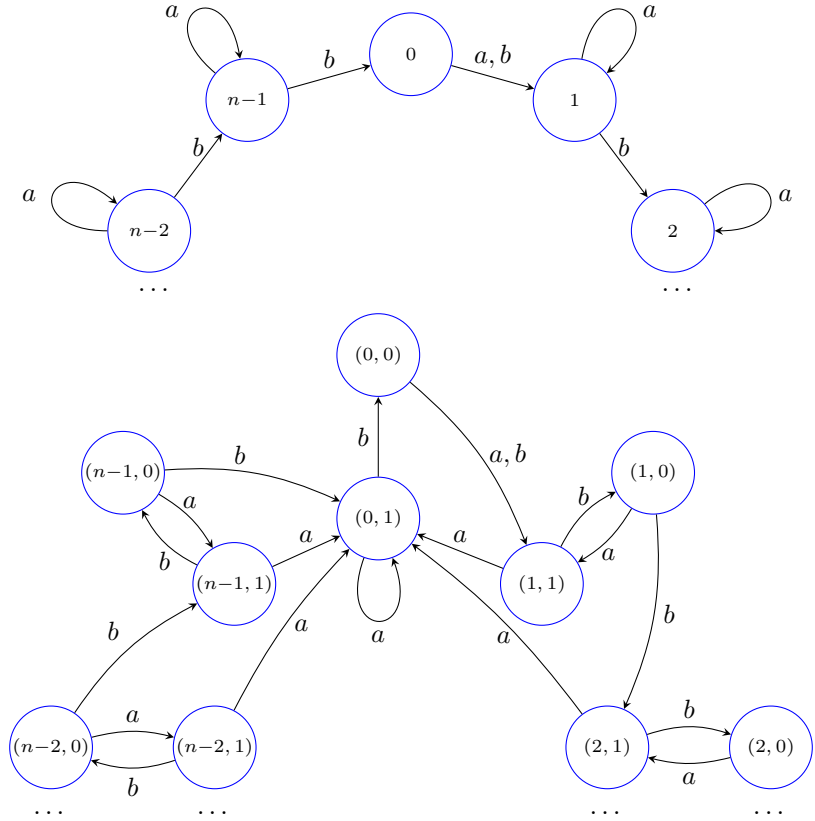
\begin{figure}[ht]
\begin{center}
\begin{tikzpicture}[>=stealth, scale=.87]

\node[draw=blue, circle, minimum size=11mm, inner sep=0pt, font=\scriptsize] (A) at (1.5,8.9) {$n{-}2$};
\node[draw=blue, circle, minimum size=11mm, inner sep=0pt, font=\scriptsize] (B) at (3,10.9) {$n{-}1$};
\node[draw=blue, circle, minimum size=11mm, inner sep=0pt, font=\scriptsize] (C) at (5.5,11.6) {0};
\node[draw=blue, circle, minimum size=11mm, inner sep=0pt, font=\scriptsize] (D) at (8,10.9) {1};
\node[draw=blue, circle, minimum size=11mm, inner sep=0pt, font=\scriptsize] (E) at (9.5,8.9) {2};

\draw[->] (D) to [out=80,in=40,looseness=8] node[right,xshift=1mm] {$a$} (D);
\draw[->] (E) to [out=40,in=0,looseness=8] node[right,xshift=0mm] {$a$} (E);
\draw[->] (B) to [out=140,in=100,looseness=8] node[left,xshift=-1mm] {$a$} (B);
\draw[->] (A) to [out=180,in=140,looseness=8] node[left,xshift=-1mm] {$a$} (A);

\draw[->] (A) -- node[above] {$b$} (B);
\draw[->] (B) -- node[above] {$b$} (C);
\draw[->] (C) -- node[above] {$a,b$} (D);
\draw[->] (D) -- node[above] {$b$} (E);

\node at (1.6,8) {\dots};
\node at (9.6,8) {\dots};

\node[draw=blue, circle, minimum size=11mm, inner sep=0pt, font=\scriptsize] (qn2) at (0,1.0)  {$ (n{-}2,0)$};
\node[draw=blue, circle, minimum size=11mm, inner sep=0pt, font=\scriptsize] (qn1) at (1.1,5.2)  {$ (n{-}1,0)$};
\node[draw=blue, circle, minimum size=11mm, inner sep=0pt, font=\scriptsize] (q0)   at (5.0,7.0)  {$ (0,0)$};
\node[draw=blue, circle, minimum size=11mm, inner sep=0pt, font=\scriptsize] (q1)   at (9.2,5.2) {$ (1,0)$};
\node[draw=blue, circle, minimum size=11mm, inner sep=0pt, font=\scriptsize] (q2)   at (11.0,1.0) {$ (2,0)$};

\node[draw=blue, circle, minimum size=11mm, inner sep=0pt, font=\scriptsize] (qn2_) at (2.5,1.0)  {$ (n{-}2,1)$};
\node[draw=blue, circle, minimum size=11mm, inner sep=0pt, font=\scriptsize] (qn1_) at (2.8,3.5)  {$ (n{-}1,1)$};
\node[draw=blue, circle, minimum size=11mm, inner sep=0pt, font=\scriptsize] (q0_)  at (5.0,4.5)  {$ (0,1)$};
\node[draw=blue, circle, minimum size=11mm, inner sep=0pt, font=\scriptsize] (q1_)  at (7.5,3.5)  {$ (1,1)$};
\node[draw=blue, circle, minimum size=11mm, inner sep=0pt, font=\scriptsize] (q2_)  at (8.5,1.0) {$ (2,1)$};

\draw[->,bend left=14] (qn2)  to node[midway, above left] {$b$} (qn1_);
\draw[->,bend left=14] (qn1)  to node[midway, above left] {$b$} (q0_);
\draw[->,bend left=14] (q0)   to node[midway, right] {$a,b$} (q1_);
\draw[->,bend left=14] (q1)   to node[midway, right] {$b$} (q2_);

\draw[->,bend left=14] (qn2)  to node[midway, above] {$a$} (qn2_);
\draw[->,bend left=14] (qn1)  to node[midway, above] {$a$} (qn1_);
\draw[->,bend left=14] (q1)   to node[midway, below] {$a$} (q1_);
\draw[->,bend left=14] (q2)   to node[midway, below] {$a$} (q2_);

\draw[->,bend left=18] (qn2_) to node[midway, below] {$b$} (qn2);
\draw[->,bend left=18] (qn1_) to node[midway, below] {$b$} (qn1);
\draw[->] (q0_)  to node[midway, left] {$b$} (q0);
\draw[->,bend left=18] (q1_)  to node[midway, above] {$b$} (q1);
\draw[->,bend left=18] (q2_)  to node[midway, above] {$b$} (q2);

\draw[->,bend left=8]  (qn2_) to node[midway, right] {$a$} (q0_);
\draw[->]  (qn1_) to node[midway, above] {$a$} (q0_);
\draw[->]  (q1_)  to node[midway, above] {$a$} (q0_);
\draw[->,bend right=8]  (q2_)  to node[midway, below] {$a$} (q0_);

\draw[->] (q0_) to [out=-110,in=-70,looseness=8] node[below] {$a$} (q0_);

\node at (0,0.0) {$\dots$};
\node at (11.0,0.0) {$\dots$};
\node at (2.5,0.0) {$\dots$};
\node at (8.5,0.0) {$\dots$};
\end{tikzpicture}
\end{center}
\caption{The automaton $\mathrsfs{C}_n$ (above) and its duplication (below), with $q_0=0$ and letter~$b$ selected}\label{fig:cerny-n}
\end{figure}

The proof of \cite[Theorem 6]{FomMarVol2013} shows that if $\mathrsfs{A}$ is a \san\ and $\ell$ is the minimum length of \sws\ for $\mathrsfs{A}$, then Alice wins in the synchronization game on $\mathrsfs{D}$ but needs more than $\ell$ moves to do so.

\v{C}ern\'{y} \cite[Lemma 1]{Cer64} proved that the shortest reset word for the automaton $\mathrsfs{C}_{n}$ is the word $(ab^{n-1})^{n-2}a$ of length $(n-1)^2$. Therefore, for the duplication of $\mC_n$, the the shortest uniform winning strategy has length at least $(n-1)^2$. By construction, this duplication has $2n$ states. Thus, we have a series of $k$-state DFAs ($k=2n$ is even) with quadratic in $k$ length of the shortest uniform winning strategy. A similar series can be constructed for odd~$k$.

We do not know whether the described lower bound is tight.

\subsection{Comparing uniform and adaptive strategies}
It is natural to expect that there exist A-automata that admit a uniform winning strategy, yet on which Alice can win in fewer moves with an adaptive strategy. In Figure~\ref{fig:adaptive-vs-uniform}, there is a simple automaton $\mE$ confirming this expectation.
\begin{figure}[hbt]
\begin{center}
\begin{tikzpicture}[>=stealth, node distance=22mm, auto]
  \tikzstyle{state}=[circle, draw=blue, minimum size=9mm]

  \node[state] (1) {1};
  \node[state, right of=1] (2) {2};
  \node[state,right of=2] (5) {0};
  \node[state, above of=5] (3) {3};
  \node[state, below of=5] (4) {4};
  \node[state, right of=5] (0) {5};

\draw[->] (1) to [out=160,in=-160,looseness=8] node[left,xshift=-1mm] {$b,c$} (1);
\draw[->] (5) to [out=115,in=155,looseness=8] node[left] {$a,b,c$} (5);
\draw[->] (3) to [out=70,in=110,looseness=8] node[above] {$b$} (3);
\draw[->] (4) to [out=-70,in=-110,looseness=8] node[below] {$c$} (4);
\draw[->] (0) to [out=20,in=-20,looseness=8] node[right] {$a$} (0);

  \path[->]
    (1) edge node {$a$} (2)
    (0) edge[bend right=18] node[swap] {$b,c$} (5)
    (2) edge[bend right=18] node[swap] {$a$} (5);

  \path[->]
    (2) edge[bend left=18] node {$b$} (3)
    (4) edge[bend left=18] node {$b$} (5)
    (4) edge[bend right=18] node[swap] {$a$} (0);

  \path[->]
    (2) edge[bend right=18] node[swap] {$c$} (4)
    (3) edge[bend left=18] node {$c$} (5)
    (3) edge[bend left=18] node {$a$} (0);
\end{tikzpicture}
\caption{The automaton $\mE$}\label{fig:adaptive-vs-uniform}
\end{center}
\end{figure}

Alice can win the synchronization game on $\mE$ in just two moves. For this, she chooses $a$ as her first move. After that, tokens remain on states 0, 2, and 5. If Bob responds with $a$, then tokens remain on states 0 and 5, and Alice wins by choosing $b$. If Bob responds with $b$ or $c$, then tokens remain on states 0 and 3 or, respectively, 0 and 4, Alice wins by choosing $c$ in the former case and $b$ in the latter.

On the other hand, no word of length 2 is a uniform winning strategy for $\mE$. To show this, we provide, for each word $xy$ with $x,y\in\{a,b,c\}$, a word $zt$ with $z,t\in\{a,b,c\}$ such that the word $xzyt$ does not reset $\mE$; see the table below, whose last row shows the states on which tokens remain after the sequence $xzyt$ of Alice and Bob's moves. Thus, whenever Alice fixes in advance a sequence $x$, $y$ of her first two moves, Bob can prevent synchronization by choosing $z$ and $t$ as his first and second responses, respectively.

\begin{center}
\begin{tabular}{p{2.6cm}||c|c|c|c|c|c|c|c|c}
Alice's moves $xy$ & $aa$ & $ab$ & $ac$ & $ba$ & $bb$ & $bc$ & $ca$ & $cb$ & $cc$\\
\hline
Bob's moves $zt$ & {\red $ba$} & {\red $ba$} & {\red $ca$} & {\red $bb$} & {\red $bb$} & {\red $bc$} & {\red $cc$} & {\red $cb$} & {\red $cc$}\\
\hline
Word $xzyt$ & $a{\red b}a{\red a}$ & $a{\red b}b{\red a}$ & $a{\red c}c{\red a}$ & $b{\red b}a{\red b}$ & $b{\red b}b{\red b}$ & $b{\red b}c{\red c}$ & $c{\red c}a{\red c}$ & $c{\red c}b{\red b}$ & $c{\red c}c{\red c}$\\
\hline
Tokens still on & 0,5 & 0,5 & 0,5 & 0,3 & 0,1,3 & 0,1 & 0,4 & 0,1 & 0,1,4
\end{tabular}
\end{center}

Inspecting Figure~\ref{fig:adaptive-vs-uniform}, one readily sees that the automaton $\mE$ is weakly acyclic, and hence, it is a DS-automaton (see Section~\ref{subsec:corollaries}). By Theorem~\ref{thm:main}, $\mE$ admits a uniform winning strategy.

In fact, the word $abc$ is a uniform winning strategy for $\mE$. To concisely present the verification of the claim, we again organize our computations in the form of a table. The following table shows all possible Bob's responses $z$ and $t$ to Alice's first two moves $a$ and $b$ and the states on which tokens remain after the sequence $azbt$ of Alice and Bob's moves. We see that the nonsingleton sets in the last row are $\{0,3\}$ and $\{0,5\}$, both sent to 0 by the action of $c$. Hence, each word of the form $azbtc$ resets $\mE$.

\begin{center}
\begin{tabular}{p{2.6cm}||c|c|c|c|c|c|c|c|c}
Bob's moves $zt$ & {\red $aa$} & {\red $ab$} & {\red $ac$} & {\red $ba$} & {\red $bb$} & {\red $bc$} & {\red $ca$} & {\red $cb$} & {\red $cc$}\\
\hline
Word $azbt$ & $a{\red a}b{\red a}$ & $a{\red a}b{\red b}$ & $a{\red a}b{\red c}$ & $a{\red b}b{\red a}$ & $a{\red b}b{\red b}$ & $a{\red b}b{\red c}$ & $a{\red c}b{\red a}$ & $a{\red c}b{\red b}$ & $a{\red c}b{\red c}$\\
\hline
Tokens still on & 0 & 0 & 0 & 0,5 & 0,3 & 0 & 0 & 0 & 0
\end{tabular}
\end{center}

Thus, there is a price to pay: while uniform strategies are easier to describe, they may not exist for some A-automata (see Lemma~\ref{lem:B2prime}), and even when they do exist, they may be less efficient. An intriguing open question is how large the speed advantage of an adaptive strategy can be.

\subsection{Synchronization game with a modified rule}

Here we modify the game rule described in Section~\ref{subsec:rules}, allowing Bob to make any finite sequence of moves in reply to each Alice's move or just skip the move. In more precise terms, once Alice completes her move in the game on a DFA $(Q,\Sigma)$, Bob can choose an arbitrary (maybe empty) word $w$ over $\Sigma$ rather than a letter in $\Sigma$ and move tokens still present on $Q$ after Alice's move according to the transformation caused by $w$. The modified rule seems to be better suited for using synchronization games to model reliably controlled systems that must perform tasks in the presence of natural or artificial interferences. Indeed, nature or an adversary does not necessarily act in the same rhythm as us, nor are they required to respond to each of our actions with exactly one obstacle or counteraction.

The modified rule looks much more favorable to Bob. It is easy to produce examples of A-automata on which Bob wins in the modified synchronization game. For instance, consider the DFA $\mF$ depicted below.
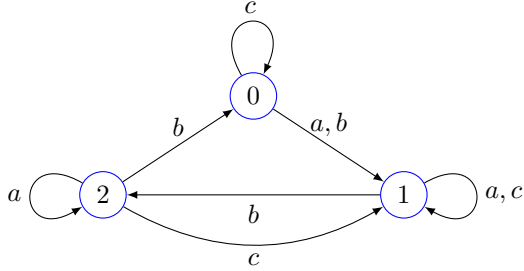
\begin{center}
\begin{tikzpicture}
	\node[fill=white, circle, draw=blue, scale=1] (0) {$0$};
	\node[fill=white, circle, draw=blue, scale=1, below = 0, xshift= -2cm, yshift= -1cm] (2) {$2$};
    \node[fill=white, circle, draw=blue, scale=1, right of = 2, xshift= 3cm] (1) {$1$};
	\draw
		(0) edge[-latex, above]  node{$a,b$} (1)
        (1) edge[-latex, below]  node{$b$} (2)
		(2) edge[-latex, bend right, below] node{$c$}(1)
		(2) edge[-latex, above]  node{$b$} (0)
		(0) edge[-latex, loop above, above, in =65, out = 120, distance = 30] node{$c$} (0)
        (1) edge[-latex, loop right, right, in =-30, out = 30, distance = 30] node{$a,c$} (1)
        (2) edge[-latex, loop left, left, in =-150, out = 150, distance = 30] node{$a$} (2);
		\end{tikzpicture}
\end{center}
Under the rule of Section~\ref{subsec:rules}, Alice wins the game on $\mF$ by choosing the letter $c$ on her first move. If Bob responds with $a$, he loses, and if he responds with $b$ or $c$, Alice wins by choosing $c$ or, respectively, $a$. Under the modified rule, however, Bob has a winning strategy on $\mF$. It consists of responding with the letter $b$ whenever Alice chooses $a$ and with the word $b^2$ whenever Alice chooses $b$ or $c$.

In view of this and similar examples, it may seem a bit surprising that the conclusion of Theorem~\ref{thm:main} persists under the modified rule. Indeed, analyzing the proof of Theorem~\ref{thm:main}, one sees that it works just as well if Bob is allowed to choose a word $w_i$ rather than a letter $b_i$ in response to Alice's $i$-th move when she spells out the word $a_1a_2\cdots a_n$. We merely apply to the word $a_1w_1a_2\cdots a_nw_n$ the argument used in the proof for the word $u=a_1b_1a_2b_2\cdots a_nb_n$. Thus, the uniform winning strategy of Theorem~\ref{thm:main} is fairly robust.

We also note that the upper bound on uniform winning strategies from Section~\ref{subsec:decide} continues to hold under the modified rule, since the configuration counting argument still applies.

\end{document}